\documentclass[12pt]{article}
\usepackage[usenames]{color}
\usepackage{graphicx, subfigure}
\usepackage{amsmath, amsthm, amssymb}
\usepackage{amsfonts}
\usepackage{fullpage}
\usepackage{ifthen}
\usepackage{url}
\usepackage[sort&compress]{natbib}
\usepackage{multirow}
\usepackage{bm}
\usepackage{tikz-qtree}
\usepackage[linesnumbered,algoruled,boxed,lined,commentsnumbered]{algorithm2e}
\usepackage[section]{placeins}

\newtheorem{theorem}{Theorem}[section]

\makeatletter
\def\@biblabel#1{}
\makeatother

\theoremstyle{plain}

\theoremstyle{definition}
\newtheorem{example}{Example}

\theoremstyle{remark}

\title{On several properties of uniformly optimal search plans}

\author{Liang Hong\footnote{Liang Hong is an Associate Professor in the Department of Mathematical Sciences,  The University of Texas at Dallas, 800 West Campbell Road, Richardson, TX 75080, USA. Tel.:~972-883-2161. Email address: liang.hong@utdallas.edu.}}
\date{\today}

\usepackage{tikz}
\usetikzlibrary{automata,arrows,positioning,calc}
\usepackage{listings}
\usepackage[labelformat=empty]{caption}
\usepackage{subfig}
\usepackage{caption}

\begin{document}

\maketitle

\begin{abstract}
The uniformly optimal search plan is a cornerstone of the optimal search theory.  It is well-known that when the target distribution is circular normal and the detection function is exponential, the uniformly search plan has several desirable properties.  This article establishes that these properties hold for any continuous target distribution.  Since there is no true target distribution, our results provide useful information to the search team when they need to choose a non-circular normal target distribution in a  real-world search  mission.

\smallskip

\end{abstract}

\section{Introduction}

The theory of optimal search for a stationary target was first developed by Koopman (1946,  1956a, b, c) and later refined by several authors such as Stone and Stanshine (1971) and Stone (1973, 1975, 1976).  It has been successfully applied in both civil and military search missions; see, for instance,  Richardson and Stone (1971), Richardson et al. (1980), Stone (1992), Kratzke et al. (2010), and Stone et al. (2014).  While recent research mainly focuses on moving targets (e.g.  Washburn 2014, Stone et al.  2016), the optimal search theory for a stationary target is still developing, as evidenced by Kadane (2015) and Clarkson et al. (2020).  More importantly, its immediate or potential value never fades over time in view of several high-profile civil and military accidents in the last decade,  such as the vanishing of AF 447 on June 1, 2009, the disappearance of Malaysia Airlines Flight 370 in 2014, the loss of Indonesia's KRI Nanggala submarine in April 2021,  and the collision of the USS Connecticut with a seamount  in South China Sea in 2021.  In addition,  the optimal search theory for moving targets is based upon its counterpart for stationary targets.  Thus, the optimal search theory for stationary targets is still relevant and important for both theory and practice.

Each search has a budget.  Koopman (1956c) first investigated maximizing the probability of finding the target once the budget has been exhausted.  In the extant literature, the uniformly optimal search plan plays a major role in the optimal search theory.  A uniformly optimal search plan is desirable because (a) it maximizes the amount of budget that could be saved (by making the search more likely to end early) and (b) it will still maximize the probability of detection if the budget is unexpectedly cut halfway through the search.  Arkin (1964) first derived sufficient conditions for existence of a uniformly optimal search policy in the Euclidean search space $\mathbb{R}^n$.  This work was extended and generalized by other authors, such as Stone (1973, 1975, 1976).

In addition, it is known that when the target distribution, which encodes belief about the target's location, is circular normal (i.e., a bivariate normal distribution with circular symmetry),  the detection function is exponential,  and the cost is proportional to effort, the uniformly optimal search plan has some appealing properties; see Section~2 for more details.  Under the assumption that the cost is proportional to allocation, this article shows that these properties hold for an arbitrary continuous target distribution when the detection function is exponential (some of these properties also hold when the detection function is regular). It is important to note that a major challenge in a search problem is the uncertainty of the target location.  A target distribution is chosen to account for this uncertainty in terms of probabilities.  Therefore, the target distribution is always subjective and there is no true target distribution.  

As pointed out in Soza and USCG (1996), the circular normal distribution has been widely employed as the target distribution in the theory of optimal search. This seems to result from the mathematical convenience it provides; that is,  many key quantities in the prior literature will admit a closed-form formula if the target distribution is taken to be circular normal. However, the target distribution is subjective, and there is no reason, theoretically or practically, to restrict our choices to the circular normal distribution.  The search team should honestly encode its knowledge and expertise into a target distribution and proceed from there.  Choosing a circular normal target distribution for mathematical convenience  will reduce the efficiency of the search as well as the likelihood of success.  For example, if the location a ship sank is known, then a circular normal distribution centered at that point could be an appropriate starting choice.  But if the current is known to drift in a certain complex pattern, then one must incorporate this prior information into the target distribution by skewing the circular normal. Hence, it is crucial to investigate whether the aforementioned  properties will still hold when the target distribution is an arbitrary continuous one.  Therefore,  our results offer important insights for any search team when they need to employ a non-circular normal target distribution in practice.  


The remainder of the article is organized as follows.  In Section~2, we set the scene  by reviewing the optimal search model and the uniformly optimal search plan.   Next,  in Section~3, we establish four general properties of the uniformly optimal search plan.  In Section~4  we provide a numerical example to illustrate our key results.   Finally,  we conclude the article with some remarks in Section~5.

\section{Model setup and notation}
We consider the problem of searching for a stationary target where the exact target location $x$ is unknown.  The Bayesian approach naturally fits here. Specifically, based on available information and professional judgment, we construct a \emph{target distribution} for $x$.  Let $\Pi$ be the (cumulative) distribution function of the target distribution. The support of $\Pi$ is often called the \emph{possibility area} and is denoted as $X$; it is the area, in our assessment,  which contains the target.  Mathematically, $X$ is a subset of the $n$-dimensional Euclidean space $\mathbb{R}^n$.  In most real-world applications,  we take $n=2$.  Intuitively, $\Pi$ reflects our knowledge prior to the search and, at the same time, accounts for our uncertainty about the target location.  Throughout,  we assume the target is stationary and we only consider the case where the target distribution is continuous with a  probability density function $\pi$.

To plan a search, we must decide how to distribute the effort in $X$.  To this end, we  let $\mathbb{R}_+\equiv [0, \infty)$ and call a function $f: X\rightarrow \mathbb{R}_+$ an \emph{allocation} on $X$ if $\int_A f(x)dx$ equals the amount of effort put in $A$ where $A$ is any subset of $X$.  Since detection is rarely perfect in any search,  we define the \emph{detection function} $d:X\times \mathbb{R}_+\rightarrow [0, 1]$ to be a function such that $d(x, y)$ denotes the \emph{conditional} probability that the target is detected if the allocation density equals $y$ at $x$ given that the target is located at $x$.  Given an allocation $f$ on $X$, the \emph{probability of detection}, denoted as $P[f]$, is given by 
\[P[f]=\int_X  d(x, f(x)) \pi(x) dx.\]
The \emph{cost function} $c(x, y): X\times \mathbb{R}_+\rightarrow \mathbb{R}_+$ stands for the cost density of applying allocation density $y$ at $x$.  For a given allocation $f$ on $X$, the cost resulting from $f$, denoted as $C[f]$, is given by 
\[C[f]=\int_Xc(x, f(x))dx.\]
We  assume $c(x, y)=y$ throughout this paper.  We define a \emph{search plan} on $X$ to be a function $\varphi: X\times \mathbb{R}_+\rightarrow \mathbb{R}_+$ such that 
\begin{enumerate}
\item[(i)]$\varphi(\cdot, t)$ is an allocation on $X$ for all $t\geq 0$;
\item[(ii)]$\varphi(x, \cdot)$ is an increasing function for all $x\in X$.
\end{enumerate}
Property~(i) says that $\varphi(\cdot, t)$ is an allocation in the first $t$ time units; Property~(ii) ensures that the allocations must ``build on each other'' as time progresses (i.e., efforts cannot be ``undone'').

In practice, the budget for conducting a search is always limited.  For a given time $T>0$, let $E(T)>0$ be the search budget available by time $T$. That is, $E(T)$ is the effort available by $T$.  Throughout, we assume the search is conducted at  speed $v$ using a sensor with a sweep width $W$ and the time available for the search is $T>0$. Hence,  $E(T)=WvT$. Intuitively,  the ``best'' search plan $\varphi^\star$, if it exists, should be the one that maximizes the probability of detection at each time $t\geq 0$, subject to the constraint $C[\varphi^\star]\leq E(T)$. In the current literature, such a search plan is called a \emph{uniformly optimal search plan}.  Precisely, we define the \emph{cumulative effort function} $E$ to be a non-negative function with domain $\mathbb{R}_+$ where $E(t)$ symbolizes the effort available by time $t$.  We always assume that $E$ is increasing.  Let $\Phi(E)$ denote the class  of search plans $\varphi$ such that 
\[\int_X\varphi(x, t)dx=E(t), \quad \text{for all $t\geq 0$}.\]
A search plan $\varphi^\star\in \Phi(E)$ is said to be \emph{uniformly optimal with $\Phi(E)$} if 
\[P[\varphi^\star(\cdot, t)]=\max\{P[\varphi(\cdot, t)]\mid \varphi\in \Phi(E)\}  \quad \text{for all $t\geq 0$}.\]
Example~2.2.9 in Stone (1975) shows that a uniformly optimal search plan does not always exist. However,  Stone (1976) establishes that there is a uniformly optimal plan if the detection function $d(x, \cdot)$ is increasing and right-continuous for each $x\in X$.  Based on the key results in Everett (1963),  a semi-closed form of the uniformly optimal plan can be derived under mild conditions using the Lagrange multiplier; see Chapter~2 of Stone (1975) or Chapter~5 of Washburn (2014) for details.  In particular, we have the following theorem. The detection function is said to be $\emph{regular}$ if $d(x, 0)=0$ and $\partial d(x, y)/\partial y$ is continuous, positive, and strictly decreasing for all $x\in X$.

\begin{theorem}[Stone 1975]
\label{thm:2.1}
If the cost function takes the form $c(x,y)=y$ for all $y\geq 0$ and $x\in X$ and the detection function is regular,  a uniformly optimal search plan $\varphi^\star$ for budget $E(T)$ can be found for any target distribution as follows.  Put
\[q_x (y)=\pi(x)\frac{\partial }{\partial y}d(x, y), \quad x\in X \text{ and $y\geq 0$},\]

\[q^{-1}_x(\lambda)=\left\{
		                           \begin{array}{ll}
		                           \text{the inverse function of $q_x(y)$  evaluated at $\lambda$},& \hbox{if $0<\lambda\leq q_x(0)$,} \\
					0, & \hbox{if $\lambda>q_x(0)$,} 
		                          \end{array}
		                         \right.\]
and 
\[Q(\lambda)=\int_X q_x^{-1}(\lambda)dx, \quad x\in X.\]
Then a uniformly optimal search plan for budget $B$ is given by $\varphi^\star(x, T)=q_x^{-1}(Q^{-1}(E(T)))$ where $Q^{-1}$ is the inverse function of $Q$. 
\end{theorem}

Here $q_x(y)$ takes into account the future change in detection and the prior probability.  Intuitively,  it represents the ``attractiveness'' of allocating more effort to the location $x$ given that $y$ amount of effort has already been allocated there, and $q^{-1}_x(\lambda)$ is the amount of allocation already required at $x$ so the attractiveness of more allocation is $\lambda$, or $0$ if the attractiveness is less than $\lambda$ before any allocation at $x$. The idea behind Theorem~\ref{thm:2.1} is as follows: distribute search effort by gradually decreasing the minimum attractiveness of $\lambda$ that is allowed into the allocation,  hence distributing effort is proportional to the most attractive areas at the current moment.  Stop when $\lambda$ falls enough so that the total allocation has reached the budget $E$.  This stopping is monitored by the function $Q$ where $Q(\lambda)$ is the total allocation if every location $x$ in the search space is allocated effort until the attractiveness of $q_x(y)$ of allocating more effort there falls to $\lambda$. 


%
\section{Several general properties of uniformly optimal search plans}
We first recall a few known facts (e.g. Stone 1975, Chapter~2).   If the target distribution is circular normal, i.e., 
\[\pi(x_1, x_2)=\frac{1}{2\pi\sigma^2}e^{-\frac{x_1^2+x_2^2}{2\sigma^2}}, \quad (x_1, x_2)\in X= \mathbb{R}^2,\]
and the detection function is exponential, i.e., $d(x, y)=1-e^{-y}$ for all $x\in X$ and $y\geq 0$,  then the uniformly optimal plan exists and takes the form (in polar coordinates)
\[\varphi^\star((r, \theta), T)=\left\{
		                           \begin{array}{ll}
		                           H\sqrt{T}-\frac{r^2}{2\sigma^2},& \hbox{if $0\leq r\leq R(T)$,} \\
					0, & \hbox{if $r>R(T)$,} 
		                          \end{array}
		                         \right.\]
where $R^2(t)=2\sigma^2H\sqrt{t}$ and $H=\sqrt{Wv/\pi\sigma^2}$.  Note that here the allocation at a point depends  only on $r$ (not on $\theta$) and that the area assigned non-zero allocation is a circle whose radius grows with $T$. These two properties result from the fact the target distribution is  circular normal,  but they do not hold for an arbitrary target distribution. 

Moreover,  the probability of detection is given by 
\[P[\varphi^\star(\cdot, T)]=1-(1+H\sqrt{T})e^{-H\sqrt{T}}, \quad T\geq 0.\]
Let $\pi_T$ be the density of the posterior probability distribution for the target location  given that the target has not been found by time $T$. Then 
\[\pi_T((r, \theta))=\left\{
		                           \begin{array}{ll}
		                           \frac{1}{2\pi\sigma^2(1+H\sqrt{T})},& \hbox{if $0\leq r\leq R(T)$,} \\
					\frac{\exp\{-\frac{r^2}{2\sigma^2}+H\sqrt{T}\}}{2\pi\sigma^2(1+H\sqrt{T})}, & \hbox{if $r>R(T)$.} 
		                          \end{array}
		                         \right.\]
It is easy to verify that this piecewise-defined function is continuous.  Also, the definition of $R(T)$ shows that $e^{-\frac{r^2}{2\sigma^2}+H \sqrt{T}}<1$ for $r>R(T)$. Thus, the second piece of $\pi_T$ is no greater than the first piece.

For $t\geq 0$, define the \emph{additional effort density} that accumulates at point $(r, \theta)$ in the time interval $[T, T+t]$ as 
\[\varphi^\star_T((r, \theta), t)=\varphi^\star((r, \theta), T+t)-\varphi^\star((r, \theta), T), \quad x\in X \text{ and $t\geq 0$}.\]
Then 
\[\varphi_T^\star((r, \theta), t)=\left\{
		                           \begin{array}{ll}
		                           H(\sqrt{T+t}-\sqrt{T}),& \hbox{if $0\leq r\leq R(T)$,} \\
					H\sqrt{T+t}-\frac{r^2}{2\sigma^2}, & \hbox{if $R(T)<r\leq R(T+t)$,} \\
					0,& \hbox{if $r>R(T+t)$.}
		                          \end{array}
		                         \right.\]

It is clear that this uniformly optimal search plan has the following nice properties:
\begin{enumerate}
\item[(i)]The detection probability $P[\varphi^\star(x, T)]$ goes to $1$ as time $T\rightarrow \infty$;
\item[(ii)]the posterior density $\pi_T$ for the target location given that the target has not been found by time $T$  is constant inside  a circle of radius $R(T)$ around the origin but not constant outside this circle;
\item[(iii)]the posterior density $\pi_T$ is flattening and spreading as $T$ increases and $\pi_T$ vanishes as $T\rightarrow\infty$; 
\item[(iv)]the additional effort that accumulates in the interval $[T, T+t]$ will be placed uniformly over the area that has been searched by time $T$. 
\end{enumerate}

Do these properties still hold when the target distribution is not circular normal? We give an affirmative answer to each of them.  Intuitively, Property~(i) says that if we keep searching, we will eventually find the target using the uniformly optimal search plan. The next theorem shows that this property holds for an arbitrary continuous target distribution when the detection function is regular.

\begin{theorem}
\label{thm:3.1}
If the target distribution is continuous with a probability density function $\pi$,  the detection function $d$ is regular,  and  $\lim_{T\rightarrow \infty} E(T)=\infty$, then 
$\lim_{T\rightarrow \infty} P[\varphi^\star(\cdot, T)]=1$.
\end{theorem}

\begin{proof}
By Theorem~\ref{thm:2.1},  a uniformly optimal search plan $\varphi^\star$ exists and is given by $\varphi^\star(x, T)=q_x^{-1}(Q^{-1}(E(T)))$. Therefore, the probability of detection is given by 
\[
P[\varphi^\star(\cdot, T)] = \int_Xd(x, \varphi^\star(x,T))\pi(x)dx =\int_{X} d(x, q_x^{-1}(Q^{-1}(E(T)))) \pi(x)dx.
\]
Since $Q$ is continuous on $(0,\infty)$ and strictly decreasing on the interval $(0, \sup\{\lambda\mid Q(\lambda)>0\})$, $Q^{-1}$ is continuous and strictly decreasing on $(0, \infty)$. In addition, note that $q_x^{-1}$ is continuous and strictly decreasing on $(0, q_x(0)]$.  Therefore,  $q_x^{-1}(Q^{-1})$ is continuous and strictly increasing on $(0, \infty)$.  Also, it is known that $\lim_{\lambda\rightarrow \infty}Q^{-1}(\lambda)=0$; see, for example, Section~2.2 of Stone (1975).  This, in view of the  assumption that $d$ is regular,  implies that $\lim_{\lambda\rightarrow \infty}q_x^{-1}(Q^{-1}(\lambda))=\infty$.  Since $|d(x, y)|\leq 1$ for all $x\in X$ and $y\geq 0$, the Dominated Convergence Theorem implies 
\begin{eqnarray*}
1\geq \lim_{T\rightarrow \infty} P[\varphi^\star(\cdot, T)] &=& \lim_{T\rightarrow \infty} \int_{X} d(x, q_x^{-1}(Q^{-1}(E(T)))) \pi(x)dx\\
&=& \int_{X}  \lim_{T\rightarrow \infty}  d(x, q_x^{-1}(Q^{-1}(E(T)))) \pi(x)dx\\
&=&  \int_{X}   d(x, q_x^{-1}(Q^{-1}(\lim_{T\rightarrow \infty}  E(T)))) \pi(x)dx\\
&=&1. 
\end{eqnarray*}
It follows that $\lim_{T\rightarrow \infty} P[\varphi^\star(\cdot, T)]=1$.
\end{proof}

To investigate Property~(ii) for an arbitrary continuous target distribution, we note that the posterior density $\pi_T$ of the target location given that the target has not been found by time $T$ can be written as 
\[
\pi_T(x)=\frac{[1-d(x, \varphi^\star(x, T))]\pi(x)}{\int_X [1-d(x, \varphi^\star(x, T))]\pi(x)dx}=\frac{[1-d(x, q_x^{-1}(Q^{-1}(E(T))) )]\pi(x)}{\int_X [1-d(x, q_x^{-1}(Q^{-1}(E(T))) )]\pi(x)dx}.
\]
The denominator of $\pi_T$ is constant because $x$ is integrated out. For the circular normal case, the condition $r\leq R(T)$ is equivalent to the condition $Q^{-1}(E(T))\leq q_x(0)$. Therefore, we need to examine two cases: (I)~$Q^{-1}(E(T))> q_x(0)$ and (II)~$Q^{-1}(E(T))\leq q_x(0)$.   Here $q_x(0)$ is the attractiveness of allocating a first bit of effort to location $x$, while $Q^{-1}(E(T))$ is the ``baseline'' of attractiveness to get into the allocation at time $T$. If $Q^{-1}(E(T))> q_x(0)$, allocation at $x$ never reaches this threshold,  so $x$ is never allocated any search effort before time $T$.  In Case~(I), the numerator of $\pi_T$ equals $[1-d(x, 0)]\pi(x)=\pi(x)$ if the detection function is regular.  Therefore, the numerator of $\pi_T$ is not constant  unless the target distribution is uniform.  For Case~(II),  the next theorem demonstrates that the numerator of $\pi_T$ (hence $\pi_T$ itself) is constant in the area $D_T=\{x\mid Q^{-1}(E(T))\leq q_x(0)\}$ if the detection function is exponential.

\begin{theorem}
\label{thm:3.2}
Suppose the target distribution is continuous with a probability density function $\pi$.
\begin{enumerate}
\item[(i)]If the detection function is of the form $d(x, y)=1-e^{-\alpha y}$ for all $x\in X$ and $y\geq 0$ where $\alpha>0$,  then the  posterior density $\pi_T$ is constant and equals $Q^{-1}(E(T))$ in the area $D_T=\{x\mid Q^{-1}(E(T))\leq q_x(0)\}$. 
\item[(ii)]If the detection function is regular and the target distribution is non-uniform,  then $\pi_T$ is not constant outside $D_T$. 
\end{enumerate}
\end{theorem}

\begin{proof}
The second statement follows from the above discussion.  Thus, it remains to establish the first statement.
In this case, we have $q_x(y)=\alpha \pi(x)e^{-\alpha y}$ and 
\[q_x^{-1}(\lambda)=\left\{
		                           \begin{array}{ll}
		                           [-\ln \lambda+\ln (\alpha \pi(x))]/\alpha,& \hbox{if $0<\lambda \leq q_x(0)$,} \\
					0, & \hbox{$\lambda>q_x(0)$.} 
		                          \end{array}
		                         \right.\]
Therefore,  the uniformly optimal plan is given by 
\[\varphi^\star(x, T)=\left\{
		                           \begin{array}{ll}
		                          [-\ln (Q^{-1}(E(T)))+\ln (\alpha \pi(x))]/\alpha,& \hbox{if $0<Q^{-1}(E(T)) \leq q_x(0)$,} \\
					0, & \hbox{$Q^{-1}(E(T))>q_x(0)$.} 
		                          \end{array}
		                         \right.\]
		                         If $0<Q^{-1}(E(T)) \leq q_x(0)$, the numerator of $\pi_T$ equals
\begin{eqnarray*}
[1-d(x, \varphi^\star(x, T))]\pi(x) 
&=& \exp\{\ln (Q^{-1}(E(T)))-\ln (\alpha \pi(x))\}\pi(x)\\
&=& Q^{-1}(E(T))/\alpha.
\end{eqnarray*}
\end{proof}
\noindent \textbf{Remark.} Statement~(i) cannot be extended to the case where the detection function is regular.  To see this,  let $d(x, y)=1-e^{-x^2y}$. Then $d(x, y)$ is regular, $q_x(y)= \pi(x) x^2 e^{-x^2y}$, and 
\[q_x^{-1}(\lambda)=\left\{
		                           \begin{array}{ll}
		                           -\frac{1}{x^2}\ln \left[\frac{\lambda}{x^2 \pi(x)}\right],& \hbox{if $0<\lambda \leq q_x(0)$,} \\
					0, & \hbox{otherwise.} 
		                          \end{array}
		                         \right.\]
Thus, when $x\in D_T$, the numerator of $\pi_T$ equals
\begin{eqnarray*}
1-\left[1- \exp\left \{-x^2\left(-\frac{1}{x^2}\ln \left[\frac{U^{-1}(E(T))}{x^2 \pi(x)}\right]\right) \right\}\pi(x)\right]
= U^{-1}(E(T))/x^2,
\end{eqnarray*}
which depends on $x$.\\

When the target distribution is circular normal, Property~(iii) shows that $\pi_T$  is flattening and spreading as $T$ increases and  $\lim_{T\rightarrow\infty}\pi_T(x)=0$ for all $x\in X$.   In fact,  this phenomenon also holds for any continuous target distribution.  Let $\phi$ denote the empty set and $\mu$ stand for the Lebesgue measure on $\mathbb{R}^n$. 

\begin{theorem}
\label{thm:3.3}
If the possibility area $X$ satisfies $\mu (X)=\infty$, the target distribution  is continuous with a probability density function $\pi$, and the detection function takes the form $d(x, y)=1-e^{-\alpha y}$ for all $x\in X$ and $y\geq 0$ where $\alpha>0$,  then $\pi_T$ is flattening and spreading as $T$ increases and $\lim_{T\rightarrow }\pi_T(x)=0$ for all $x\in X$. 
\end{theorem}

\begin{proof}
Put $D_T=\{x\in X\mid Q^{-1}(E(T))\leq q_x(0)\}$. Then the proof of Theorem~\ref{thm:3.2} implies $q_x(0)=\alpha\pi(x)$ and 
\begin{eqnarray}
\label{eq:post}
\pi_T(x) &=&  \frac{[1-d(x, \varphi^\star(x, T))]\pi(x)}{\int_X [1-d(x, \varphi^\star(x, T))]\pi(x)dx} 
	= \left\{
		                         \begin{array}{ll}
		                           \frac{Q^{-1}(E(T))/\alpha}{\int_X [1-d(x, \varphi^\star(x, T))]\pi(x)dx},& \hbox{if $x\in D_T$,} \\
					 \frac{\pi(x)}{\int_X [1-d(x, \varphi^\star(x, T))]\pi(x)dx}, & \hbox{$x\in D_T^c$.} 
		                          \end{array}
		                         \right.
\end{eqnarray}
Thus, when $x\in D_T^c$,  we have $Q^{-1}(E(T))/\alpha<\pi(x)$,  showing that the part of $\pi_T$ on $D_T$ is no greater than the part of $\pi_T$ on $D_T^c$.  Since $D^c_T=\{x\in X\mid Q^{-1}(E(T))> q_x(0)\}\rightarrow \phi$ as $T\rightarrow\infty$, it suffices to show that $\lim_{T\rightarrow \infty}\pi_T(x)=0$ for  $x\in D_T$.  It follows from the proof of Theorem~\ref{thm:3.2} that
\[[1-d(x, \varphi^\star(x, T))]\pi(x)=
\left\{
		                         \begin{array}{ll}
		                          Q^{-1}(E(T))/\alpha,& \hbox{if $x\in D_T$,} \\
					1, & \hbox{$x\in D_T^c$.} 
		                          \end{array}
		                         \right.
\]
Therefore, for $x\in D_T$,  we have 
\begin{eqnarray*}
0\leq \pi_T(x) 
&=& \frac{Q^{-1}(E(T))/\alpha}{\int_{D_T} Q^{-1}(E(T))/\alpha dx+ \int_{D_T^c}\pi(x)dx}\\
&\leq& \frac{Q^{-1}(E(T))}{\int_{D_T} Q^{-1}(E(T))dx}= \frac{Q^{-1}(E(T))}{Q^{-1}(E(T)) \int_{D_T} 1\ dx}= \frac{1}{\mu (D_T)}.
\end{eqnarray*}

To finish the proof, it remains to show that $\lim_{T\rightarrow \infty} \mu(D_T)=\infty$.  To this end, note that $Q^{-1}$ is continuous and strictly decreasing.  As $\lim_{\lambda\rightarrow \infty}Q(\lambda)=0$ (e.g.  Stone 1975, p.48), we know $\lim_{T\rightarrow \infty}Q^{-1}(E(T))=0$.  (This is consistent with the interpretation of $Q^{-1}(E(T))$ as the baseline of attractiveness  we mentioned before Theorem~\ref{thm:3.2}.) Since $\int_X\pi(x)dx=1<\infty$,  we have $\lim_{||x||\rightarrow \infty} \pi (x)=0$ where $||\cdot||$ is the Euclidean norm on $\mathbb{R}^n$.  Therefore, the continuity of the Lebegsue measure implies 
\begin{eqnarray*}
\lim_{T\rightarrow \infty} \mu (D_T) &\leq& \lim_{T\rightarrow \infty}  \mu (\{x\mid ||x||\leq \inf\{||a||\mid \pi(a)=Q^{-1}(E(T))\}\}) \\
&=&  \mu \left( \lim_{T\rightarrow \infty} \{x\mid ||x||\leq \inf\{||a||\mid \pi(a)=Q^{-1}(E(T))\}\} \right)\\
&=& \infty.
\end{eqnarray*}
\end{proof}
\noindent \textbf{Remark.} Theorem~\ref{thm:3.3} does not extend to a regular detection function. To see this, consider  the
example in Remark of Theorem~\ref{thm:3.2}.  In this case,  if $x\in D_T$, then
\begin{eqnarray*}
0\leq \pi_T(x) 
		     &=& \frac{Q^{-1}(M(T))/x^2}{\int_X Q^{-1}(M(T))/x^2 dx}=\frac{1/x^2}{\int_X 1/x^2 dx},\quad \text{for all $T>0$.  }
\end{eqnarray*}
Therefore, for any $T>0$ we have
\[\pi_T(x)=
\left\{
		                           \begin{array}{ll}
		                           0,& \hbox{if the the possibility area $X$ contains the origin,} \\
					 \text{a positive constant}, & \hbox{otherwise.} 
		                          \end{array}
		                         \right.
\]

Property~(iv) says that a uniformly optimal search plan will place the additional effort that accumulates in the interval $[T, T+t]$ uniformly over the area that has been searched by time $T$.  It turns out this property  holds for an arbitrary continuous target distribution too.

\begin{theorem}
\label{thm:3.4}
If the target distribution is  continuous with a density function $\pi$ and the detection function has the form $d(x, y)=1-e^{-\alpha y}$ for all $x\in X$ and $y\geq 0$ where $\alpha>0$,  then the  uniformly optimal search plan will place the additional effort uniformly over the area that has been searched by time $T$.
\end{theorem}

\begin{proof}
From the proof of Theorem~\ref{thm:3.2}, we know that the uniformly optimal search plan $\varphi^\star$ is given by 
\[\varphi^\star(x, T)=\left\{
		                           \begin{array}{ll}
		                           [-\ln (Q^{-1}(E(T)))+\ln (\alpha \pi(x))]/\alpha,& \hbox{if $0<Q^{-1}(E(T)) \leq q_x(0)$,} \\
					0, & \hbox{if $Q^{-1}(E(T))>q_x(0)$.} 
		                          \end{array}
		                         \right.\]
Let $\varphi^\star_T(x, t)=\varphi^\star(x, T+t)-\varphi^\star(x, T)$ for $x\in X$ and $t\geq 0$ be the additional effort density. It follows that $\varphi^\star_T(x, t)$ equals
\[\left\{
		                           \begin{array}{ll}
		                           -\frac{1}{\alpha} [\ln (Q^{-1}(E(T+t))/Q^{-1}(E(T)))],& \hbox{if $0<Q^{-1}(E(T+t))< Q^{-1}(E(T))\leq q_x(0)$,} \\
		                           \frac{1}{\alpha}[-\ln (Q^{-1}(E(T+t)))+\ln (\alpha \pi(x))],& \hbox{if $0<Q^{-1}(E(T+t)) \leq q_x(0) < Q^{-1}(E(T))$,} \\
					0, & \hbox{if $Q^{-1}(E(T+t))>q_x(0)$,} 
		                          \end{array}
		                         \right.\]
for all $t>0$. Therefore,  the theorem follows.
\end{proof}

\noindent \textbf{Remark.} When $d$ is regular,  the uniformly optimal search plan $\varphi^\star$ is given by 
\[\varphi^\star(x, T)=\left\{
		                           \begin{array}{ll}
		                           q_x^{-1}(Q^{-1}(E(T))),& \hbox{if $0<Q^{-1}(E(T)) \leq q_x(0)$,} \\
					0, & \hbox{if $Q^{-1}(E(T))>q_x(0)$,} 
		                          \end{array}
		                         \right.\]
and $\varphi^\star_T(x, t)$ equals
\[\left\{
		                           \begin{array}{ll}
		                           q_x^{-1}(Q^{-1}(E(T+t)))-q_x^{-1}(Q^{-1}(E(T))),& \hbox{if $0<Q^{-1}(E(T+t))< Q^{-1}(E(T))\leq q_x(0)$,} \\
		                           q_x^{-1}(Q^{-1}(E(T+t))),& \hbox{if $0<Q^{-1}(E(T+t)) \leq q_x(0) < Q^{-1}(E(T))$,} \\
					0, & \hbox{if $Q^{-1}(E(T+t))>q_x(0)$,} 
		                          \end{array}
		                         \right.\]
In generally, $q_x^{-1}(Q^{-1}(E(T+t)))-q_x^{-1}(Q^{-1}(E(T)))$ might depend on $x$.  For example, if we consider the example in Remark of Theorem~\ref{thm:3.2}, then $\varphi^\star_T(x, t)$ equals
\[\left\{
		                           \begin{array}{ll}
		                           -\frac{1}{x^2} [\ln (Q^{-1}(E(T+t))/Q^{-1}(E(T)))],& \hbox{if $0<Q^{-1}(E(T+t))< Q^{-1}(E(T))\leq q_x(0)$,} \\
		                           \frac{1}{x^2}[-\ln (Q^{-1}(E(T+t)))+\ln (x^2 \pi(x))],& \hbox{if $0<Q^{-1}(E(T+t)) \leq q_x(0) < Q^{-1}(E(T))$,} \\
					0, & \hbox{if $Q^{-1}(E(T+t))>q_x(0)$,} 
		                          \end{array}
		                         \right.\]
for all $t>0$.  Therefore, Theorem~\ref{thm:3.4} does not extend to the case of a regular detection function.  

\section{Example}
\begin{example}
\label{example1}
Suppose a search team is tasked with searching for a sunk ship.   It is known that the current has been drifting towards north or/and east since the ship sank.  The teams decides to choose the density of the target distribution to be
\[\pi(x)=\pi(x_1, x_2)=\left\{
		                           \begin{array}{ll}
		                           e^{-(x_1+x_2)},& \hbox{$x_1>0$ and $x_2>0$,} \\
					0, & \hbox{otherwise,} 
		                          \end{array}
		                         \right.\]
where $x=(x_1, x_2)\in X=\mathbb{R}^2$.  In addition, the team takes the detection function to be $d(x,  y)=1-e^{-y}$ for $x\in X$ and $y\geq 0$.   In this case,  $q_x(y)=\pi(x)e^{-y}$ and 
\[q_x^{-1}(\lambda)=\left\{
		                           \begin{array}{ll}
		                           -\ln \lambda-(x_1+x_2),& \hbox{if $0<\lambda\leq e^{-(x_1+x_2)}$,} \\
					0, & \hbox{otherwise,} 
		                          \end{array}
		                         \right.\]                    
Therefore, $Q(\lambda)=\int_0^{-\ln \lambda}\int_0^{(-\ln \lambda)-x_1} [(-\ln \lambda)-x_1-x_2]dx_1dx_2=-\frac{1}{6}(\ln \lambda)^3$ and the uniformly optimal search plan given by 
\begin{equation}
\label{eq:unif}
\varphi^\star((x_1, x_2), T)=\left\{
		                           \begin{array}{ll}
		                           (6WvT)^{1/3}-(x_1+x_2),& \hbox{if $0< x_1+x_2\leq (6WvT)^{1/3}$,} \\
					0, & \hbox{otherwise,} 
		                          \end{array}
		                         \right.
\end{equation}
where $T$ is the time available for search.  For  $W=1$, $v=5$, and two values of $T$, Figure~1 provides a 3-D plot of the uniformly optimal search plan $\varphi^\star$ given by (\ref{eq:unif}).


\begin{figure}[h]
\begin{center}
\subfigure[$T=10$]{\scalebox{0.40}{\includegraphics{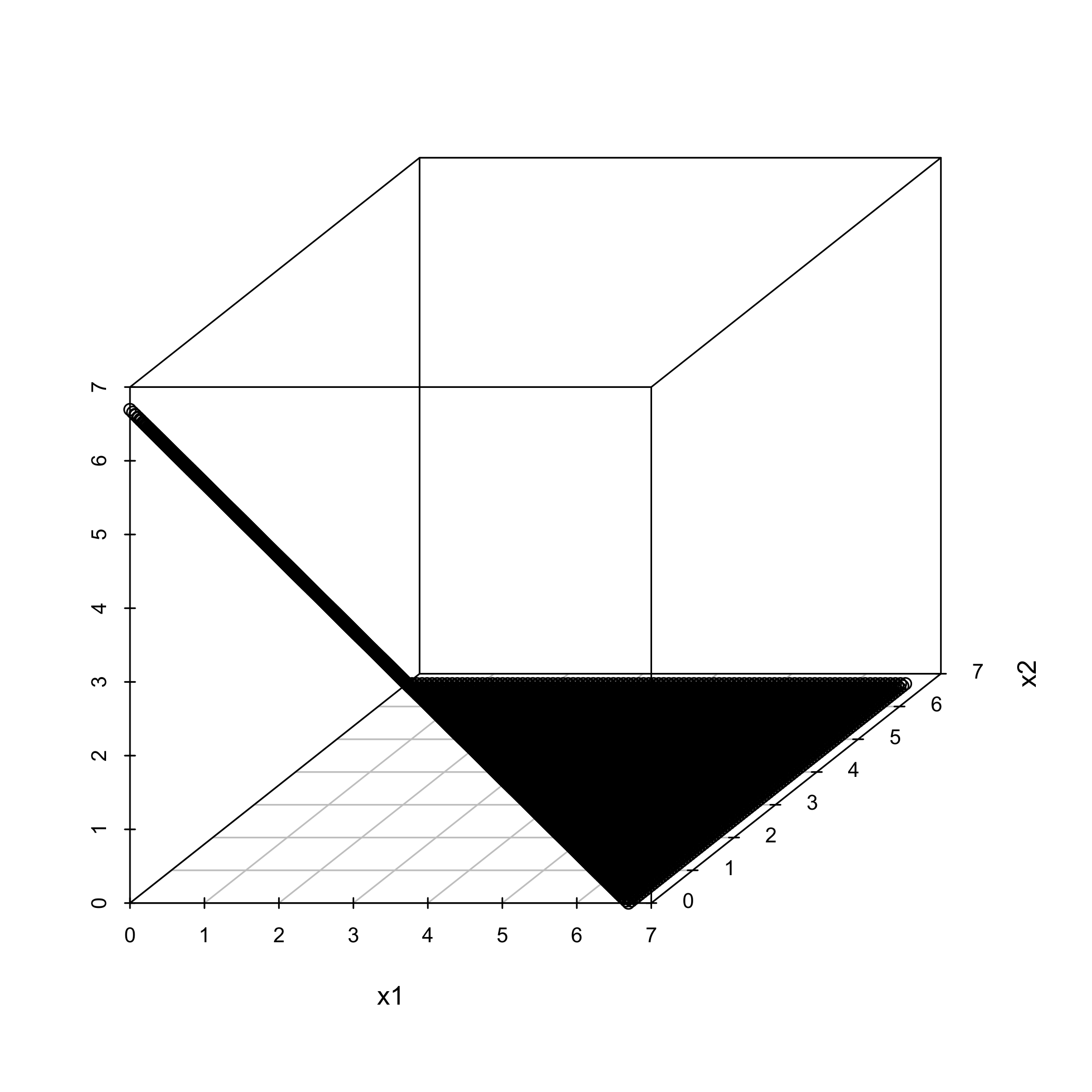}}}
\subfigure[$T=50$]{\scalebox{0.40}{\includegraphics{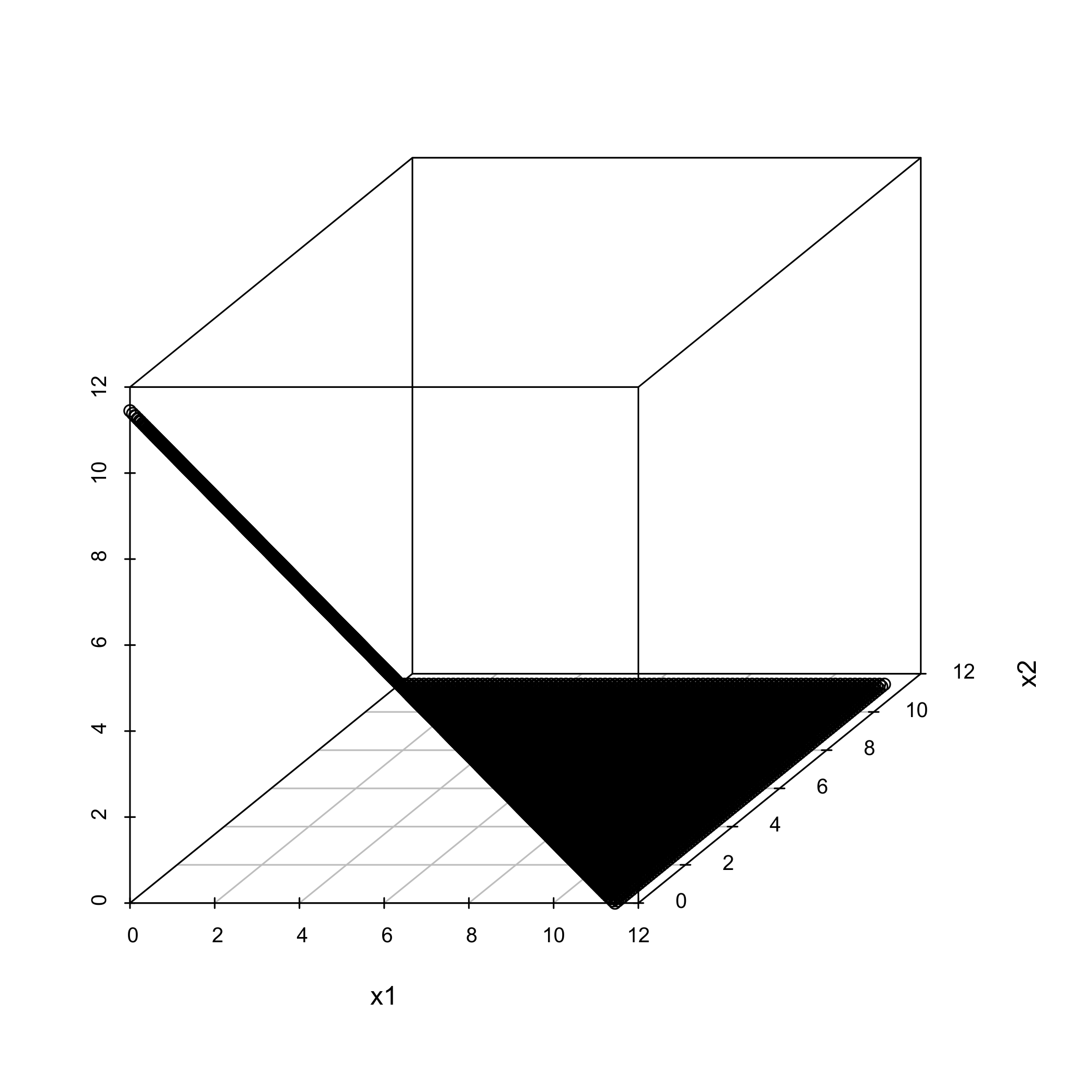}}}
\caption{Figure~1:\ 3-D plots of the uniformly optimal search plan $\varphi^\star$ in Example~1 when $W=1$ and $v=5$ for  $T=10$ and $T=50$.}
\label{fig:fitted}
\end{center}
\end{figure}

Put $D_T=\{(x_1, x_2)\mid x_1+x_2\leq (6WvT)^{1/3}\}$. Then the probability of detection equals
\begin{eqnarray*}
P[\varphi^\star(\cdot, T)] &=& \int_X d(x, \varphi^\star(x,T))\pi(x)dx\\
&=& \int_{D_t} [1-e^{(-(6WvT)^{1/3}+(x_1+x_2))}] e^{-(x_1+x_2)} dx_1 dx_2 \\
&=& 1-e^{-(6WvT)^{1/3}}-(6WvT)^{1/3}e^{-(6WvT)^{1/3}}-(1/2)(6WvT)^{2/3}e^{-(6WvT)^{2/3}}.
\end{eqnarray*}
Thus,  $\lim_{T\rightarrow \infty} P[\varphi^\star(x, T)] =1$.

By (\ref{eq:unif}), we know  $\pi(x)[1-d(x, \varphi^\star (x, T))]=e^{-(6Wvt)^{1/3}}$ for $x\in D_T$. Also,  $m(D_T)=(1/2)(6Wvt)^{2/3}$ and 
\begin{eqnarray*}
\int_{D^c_T} \pi(x)dx
&=&1-\int_0^{(6Wvt)^{1/3}}\int_0^{(6Wvt)^{1/3}-x_1}e^{-(x_1+x_2)}dx_2 dx_1
=[1+(6Wvt)^{1/3}]e^{-(6Wvt)^{1/3}}.
\end{eqnarray*}
It follows from (\ref{eq:post}) that
\begin{eqnarray*}
\pi_T(x_1, x_2) &=& \left\{
		                         \begin{array}{ll}
		                          \frac{1}{1+(6WvT)^{1/3}+1/2(6WvT)^{2/3}},& \hbox{if $0< x_1+x_2\leq (6WvT)^{1/3}$,} \\
				        \frac{e^{-(x_1+x_2)}}{1+(6WvT)^{1/3}+1/2(6WvT)^{2/3}}, & \hbox{if $x_1+x_2>(6WvT)^{1/3}$.} 
		                          \end{array}
		                         \right.
\end{eqnarray*}
It is evident that $D^c_T=\{(x_1, x_2)\mid x_1+x_2> (6WvT)^{1/3}\}\rightarrow \phi$ as $T\rightarrow \infty$ and that the part of $\pi_T$ on $D_T^c$ is less than the part of $\pi_T$ on $D_T$. On the area $D_T$, the posterior density $\pi_T$ is constant and tends to zero as $T\rightarrow \infty$. Therefore, $\pi_T(x)\rightarrow 0$ as $T\rightarrow \infty$ for all $x\in X$.  Figure~2 provides 3-D plots of the posterior density $\pi_T$ for  $W=1$, $v=5$, and two values of $T$. As we can see, the posterior density $\pi_T$ is flattering and spreading as $T$ increases. 
\begin{figure}[h]
\begin{center}
\subfigure[$T=10$]{\scalebox{0.20}{\includegraphics{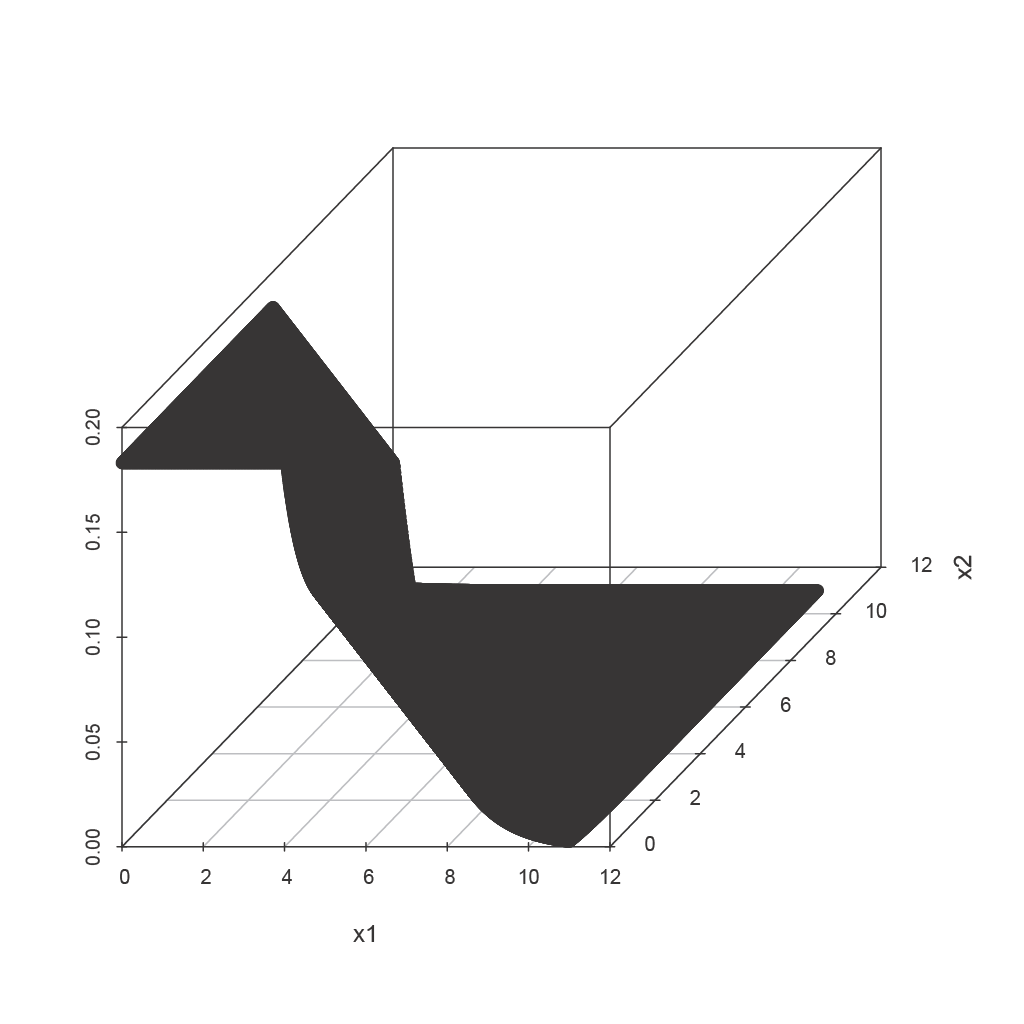}}}
\subfigure[$T=50$]{\scalebox{0.20}{\includegraphics{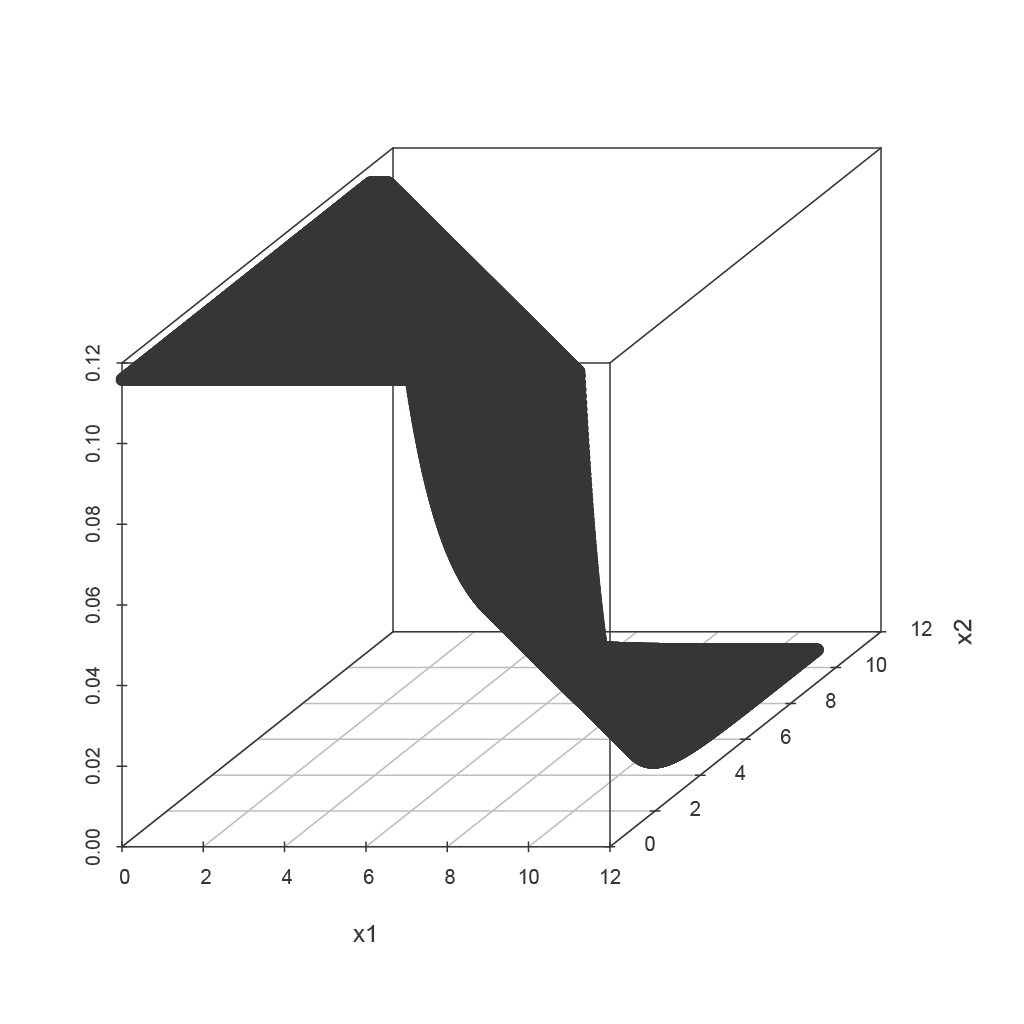}}}
\caption{Figure~2:\ 3-D plots of the posterior density $\pi_T$ outside $D_T$ in Example~1 when $W=1$ and $v=5$ for $T=10$ and $T=50$.}
\label{fig:fitted}
\end{center}
\end{figure}

In this case, the additional effort density equals
\[\varphi^\star_T((x_1, x_2), t)=\left\{
		                           \begin{array}{ll}
		                          \  [6Wv(T+t)]^{1/3}-(6WvT)^{1/3},& \hbox{if $0< x_1+x_2\leq (6WvT)^{1/3}$,  } \\
		                           \ [6Wv(T+t)]^{1/3}-(x_1+x_2),& \hbox{if $(6WvT)^{1/3}< x_1+x_2\leq [6Wv(T+t)]^{1/3}$,} \\
					\ 0, & \hbox{if $x_1+x_2> [6Wv(T+t)]^{1/3}$.} 
		                          \end{array}
		                         \right.\]
Therefore, $\varphi_T^\star$ is uniform over the area that has been searched by time $T$.
\end{example}


\section{Concluding remarks}
The uniformly optimal search plan plays a vital role in the theory of optimal search.  In the existing literature as well as many applications,  a circular normal target distribution is often employed for mathematical convenience.  In such a case,  a uniformly optimal search plan has been proved to possess several desirable properties.  To improve search efficiency and increase the probability of success, a search team should choose whatever target distribution that most closely reflects the prior information. Therefore, it is important to investigate whether these properties still hold for an arbitrary continuous target distribution. This article confirms that it is the case.

\section*{Acknowledgments}
I thank the Editor-in-Chief,  the Associate Editor, and two anonymous reviewers for many useful comments and suggestions.

\section*{References}

\begin{description}


\item{} Arkin, V.I.~(1964). Uniformly optimal strategies in search problems. \emph{Theory of Probability and its Applications}~9(4): 647--677. 

\item{} Clarkson, J., Glazebrook, K.D.~and Lin, K.Y.~(2020). Fast or slow: search in discrete locations with two search models. \emph{Operations Research}~68(2), 552--571.

\item{} Everett, H.~(1963). Generalized Lagrange multiplier method for solving probglems of optimum allocation of resources. \emph{Operations Research}~11, 399--417.

\item{} Kadane, J.B.~(2015). Optimal discrete search with technological choice.  \emph{Mathematical Methods of Operations Research}~81, 317--336. 

\item{} Koopman, B.O.~(1946). Search an screening. \emph{Operations Evaluation Group Report No. 56 (unclassified).} Center for Naval Analysis, Rosslyn, Virginia. 

\item{} Koopman, B.O.~(1956a). The theory of search, I. Kinematic bases. \emph{Operations Research}~4, 324--346.

\item{} Koopman, B.O.~(1956b). The theory of search, II. Target detection. \emph{Operations Research}~4, 503--531.

\item{} Koopman, B.O.~(1956a). The theory of search, III The optimum distribution of seraching efforts. \emph{Operations Research}~4, 613--626. 

\item{} Kratzke, T.M., Stone, L.D., and Frost J.R.~(2010). Search and rescue optimal planning system.   \emph{Proceedings of the 13th International Conference on Information Fusion, Edinburgh, UK, July 2010}, 26--29.

\item{} Richardson, H.R.~and Stone, L.D.~(1971). Operations analysis during the underwater search for Scorpion. \emph{Naval Research Logistic Quarterly}~18, 141--157. 

\item{} Ricardson, H.R., Wagner D.H.~and Discenza, J.H.~(1980). The United States Coast Guard Computer-assisted Search Planning System (CASP). \emph{Naval Research Logistic Quarterly}~27,  659--680. 

\item{} Soza Co. Ltd and U.S. Coast Guard~(1996).  \emph{The Theory of Search: A Simplified Explanation}.  U.S. Coast Guard: Washington, D.C..

\item{} Stone, L.D.~(1973). Totally optimality of incrementally optimal allocations. \emph{Naval Research Logistics Quarterly}~20, 419--430.

\item{} Stone, L.D.~(1975). \emph{Theory of Optimal Search}.  Academic Press: New York.

\item{} Stone, L.D.~(1976). Incremental and total optimization of separable functionals with constraints. \emph{SIAm Journal on Control and Optimization}~14,  791--802.

\item{} Stone, L.D., Royset, J.O., and Washburn, A.R.~(2016). \emph{Optimal Search for Moving Targets}. Springer: New York.

\item{} Stone, L.D.~and Stanshine J.A.~(1971). Optimal searching using uninterrupted contact investigation.  \emph{SIAM Journal of Applied Mathematics}~20, 241--263.

\item{} Stone, L.D.~(1992). Search for the SS Central America: mathematical treasure hunting. \emph{Interfaces}~22: 32--54. 

\item{} Stone, L.D., Keller, C.M. , Kratzke, T.M.~and Strumpfer, J.P.~(2014). Search for the wreckage of Air France AF 447. \emph{Statistical Science}~29:69--80. 


\item{} Washburn, A.~(2014). \emph{Search and Detection}, 5th Edition. Create Space: North Carolina.

\end{description}

\end{document}